\documentclass[11pt]{article}

\usepackage[T1]{fontenc}
\usepackage[utf8]{inputenc}
\usepackage[margin=1in]{geometry}
\usepackage{amsmath,amssymb,amsthm}
\usepackage{newtxtext,newtxmath}
\usepackage{booktabs}
\usepackage{array}
\usepackage{tabularx}
\newcolumntype{Y}{>{\raggedright\arraybackslash}X}
\usepackage{enumitem}
\usepackage{microtype}
\usepackage{xcolor}
\usepackage{float}
\usepackage{placeins}
\usepackage{tikz}
\usepackage{hyperref}
\usepackage[numbers,sort&compress]{natbib}
\usetikzlibrary{arrows.meta,positioning,calc,decorations.pathreplacing,patterns,shapes.geometric}

\definecolor{spData}{RGB}{0,114,178}
\definecolor{spForward}{RGB}{0,158,115}
\definecolor{spBackward}{RGB}{213,94,0}
\definecolor{spCallback}{RGB}{204,121,167}
\definecolor{spWait}{RGB}{230,159,0}
\definecolor{spOther}{RGB}{86,180,233}
\definecolor{spInk}{RGB}{40,40,40}

\newcolumntype{L}[1]{>{\raggedright\arraybackslash}p{#1}}
\setlist[itemize]{leftmargin=1.35em}
\setlist[enumerate]{leftmargin=1.6em}

\hypersetup{
  colorlinks=true,
  linkcolor=blue!50!black,
  citecolor=blue!50!black,
  urlcolor=blue!50!black,
  pdftitle={StageFrontier: Synchronization-Aware Stage Accounting for Distributed ML Training},
  pdfauthor={Boram Yoon, Wei Chen, Ville Kallioniemi},
  pdfkeywords={distributed training, PyTorch, profiling, stragglers, synchronization, performance}
}
\newcommand{\sys}{StageFrontier}
\newcommand{\ranks}{\mathcal{R}}

\newcommand{\steps}{\mathcal{T}}
\newcommand{\argmaxop}{\operatorname*{arg\,max}}
\newcommand{\median}{\operatorname{median}}

\theoremstyle{definition}

\newtheorem{proposition}{Proposition}

\newtheorem{theorem}{Theorem}

\title{\sys{}: Synchronization-Aware Stage Accounting for Distributed ML Training}

\author{
  Boram Yoon \\ NVIDIA \\ \texttt{byoon@nvidia.com}
  \and
  Wei Chen \\ NVIDIA \\ \texttt{weich@nvidia.com}
  \and
  Ville Kallioniemi \\ NVIDIA \\ \texttt{vkallioniemi@nvidia.com}
}

\date{}

\begin{document}
\maketitle

\begin{abstract}
When a distributed training job slows down, the hard part is knowing where to look. Synchronization hides the cause: a stall on one rank shows up as a wait on the others, so a data delay on a single rank can surface as backward time across the group. The cheap dashboards that run all the time --- per-stage averages and maxima --- misread this, double-counting the same exposed delay or burying the slow rank in an average, while full profilers see it clearly but are far too heavy to leave on.

\sys{} is an always-on signal that closes this gap. Each rank reports only a short ordered vector of coarse stage durations --- data, forward, backward, and so on --- timed with CPU wall-clock, with no synchronized clocks and no kernel tracing. At each stage boundary, \sys{} takes the cumulative time of whichever rank is furthest along; the increments of this \emph{frontier} form an exact, additive accounting of the step's exposed time and point to the stage and rank where group-visible delay first appears, telling an operator where to aim a heavy profiler, not which fix to make. The accounting is exact, but the coarse signal alone cannot tell whether a leading stage truly \emph{caused} the slowdown or merely ran alongside it; \sys{} labels the windows where that distinction needs more evidence instead of guessing.

A PyTorch implementation adds under $0.2\%$ throughput overhead through 128 ranks on Gloo and NCCL, places injected faults among its top two suspects on all 50 rows of a hidden-rank DDP test, and recovers the same top-stage routing as PyTorch Profiler, HTA, and Nsight Systems once their traces are reduced to the same coarse stages --- from a $0.11$\,MB summary instead of a $15.81$\,GB trace.
\end{abstract}

\section{Introduction}

When a distributed training job is slow, the first question is where the delay actually originated, not just where it appears in stage timers. Synchronization displaces visible symptoms: a slow data stage on one rank surfaces as backward wait on the others, so the stage with the largest timer is often downstream from the real cause. The \emph{exposed} delay --- the portion not hidden under overlap, which genuinely adds to step time --- could be anywhere in the pipeline: data loading, forward/loss, backward, callbacks, or optimizer work, on any rank. A wrong first answer is expensive: operators collect heavy traces, drain nodes, or rerun a multi-thousand-GPU job before discovering the symptom was a synchronization-displaced reflection of an upstream stall on a different rank. Modern systems make this harder: DDP overlaps gradient communication with backward through CUDA streams, FSDP-style sharding inserts all-gather and reduce-scatter boundaries that absorb sibling skew, and pipeline, tensor, and expert parallelism add role-specific dependencies that mix into one global rank index~\cite{li2020pytorchdistributed,pytorchddp,rajbhandari2020zero}. Gradient accumulation hides which microbatch waited, and asynchronous CUDA execution means CPU wall-clock time records when work \emph{becomes host-visible}, not where it launched~\cite{pytorchcuda}. A rank that looks slow in backward may simply be waiting for a sibling whose input pipeline stalled; scalar stage timers alone cannot separate these executions.

\paragraph{Existing tools sit at the two ends of a cost/coverage curve.}
\emph{Throughput dashboards} that aggregate per-stage maxima or averages run continuously but are ambiguous on displaced work: maxima can add an upstream stall and a downstream duplicate wait, averages hide rare rank tails, and a slowest-rank breakdown can blame a victim. \emph{Heavyweight profilers}~\cite{pytorchprofiler,tensorflowprofiler,nsightsystems} and trace analyzers~\cite{holistictraceanalysis} resolve this ambiguity by directly observing waits, idle gaps, and timeline dependencies, but at much higher overhead and storage cost; HPC critical-path and wait-state analyzers~\cite{geimer2008scalasca,boehme2012criticalpath,hu2022dpro,sridharan2023chakra} produce richer attribution still but require event traces or full operator graphs. The missing layer is a \emph{profiler-router}: low enough volume to run continuously, structured enough to tell the operator which rank, stage, and role to investigate first, and labeled clearly enough that downstream automation can decide \emph{when} to invoke the heavier tools.

\paragraph{Frontier accounting closes the gap.}
The key insight is simple: at each stage boundary, track how far the \emph{furthest-along} rank has progressed. The increments of this running maximum form an exact, additive decomposition of the step's exposed time, and each increment points to the stage where group-visible delay first appeared --- with no synchronized clocks and no per-stage \texttt{torch.cuda.synchronize()} (see Section~\ref{sec:method} and Figure~\ref{fig:frontier-accounting}). Where heavyweight profilers observe idle gaps, synchronization events, and dependency edges directly, \sys{} sees only per-rank stage durations in which any wait is mixed into the stage where the host observed it. From how those progress curves cross, it recovers the same decomposition exactly --- cheaply enough to run all the time. Bottleneck-\emph{cause} attribution additionally requires an explicit synchronization-wait exposure model, which holds most cleanly for homogeneous synchronous DDP-style training. Outside that model (MoE, model-parallel roles, heavy overlap), \sys{} emits a downgrade label --- \texttt{co\_critical} (the lead stage may merely run \emph{alongside} the bottleneck rather than cause it), \texttt{role\_aware\_needed} (ranks play different parallelism roles so global aggregation is unsafe), or \texttt{telemetry\_limited} (the signal is incomplete) --- and leaves the cause open.

\paragraph{Contributions.}
\begin{enumerate}
\item A \emph{minimal telemetry contract} for ordered, residual-closed (stage durations sum back to the measured step time), clock-independent distributed stage vectors collected without per-stage CUDA synchronization.
\item \emph{Frontier accounting} with exact telescoping, a slack identity that prevents downstream double-charging, and formal $\min(R,S)$ overcounting / $R$ undercounting bounds against max/average summaries.
\item \emph{Evidence semantics} that distinguish accounting, synchronization-wait attribution, direct exposure, co-criticality, and telemetry limits.
\item A \emph{PyTorch implementation and confirming evaluation}: compact fault routing (injected stage lands in the top-two candidate stages on all 50 test rows; metric defined in Section~\ref{sec:sync-wait}) at sub-percent always-on overhead through 128 ranks on Gloo and NCCL, broad-stage agreement with PyTorch Profiler, HTA (Holistic Trace Analysis), and Nsight on a selected window, and removed-injection, gradient-accumulation, and FSDP/ZeRO-1 spot checks.
\end{enumerate}
Five field cases (Section~\ref{sec:cases}) and a recorded set of failure modes (Section~\ref{sec:failure-modes}) show how the routing signal maps to operator action and where it breaks.

\section{Motivating Example}
\label{sec:motivating}

Figure~\ref{fig:sync-displacement} shows the displacement pattern that motivates the accounting model. Rank $r_0$ reaches its data boundary late; the other two ranks arrive at the backward synchronization point early and observe the same group delay as backward wait. Host-visible scalar timers therefore report a large \emph{backward} time on $r_1$ and $r_2$ even though no rank performed extra backward work --- the slowdown originated in $r_0$'s data stage. On the host-visible durations $r_0=(6.0,1.0,1.2)$, $r_1=(1.0,1.0,6.2)$, $r_2=(1.1,1.0,6.0)$ for (data, fwd, bwd), per-stage maximum totals $6.0+1.0+6.2=13.2$\,s, selecting the upstream data maximum and also a duplicate downstream backward wait; per-stage average smears the duplicates while hiding $r_0$'s tail; the frontier (the running maximum over ranks of cumulative stage time; defined precisely in Section~\ref{sec:method}) charges $6.0$ to data, $1.0$ to forward, $1.2$ to backward, summing to exactly the exposed makespan $8.2$\,s and pointing to the upstream boundary that first exposes it.

\begin{figure}[t]
\centering
\begin{tikzpicture}[x=1.15cm,y=0.92cm,font=\footnotesize,>=Latex]
\tikzset{
  productive/.style={draw=black!55,line width=0.25pt},
  waitspan/.style={draw=black!55,line width=0.25pt,fill=spWait!18,pattern=north east lines,pattern color=black!45},
  markline/.style={black!75,line width=0.55pt,densely dashed},
  callout/.style={->,>=Latex,line width=0.4pt,black!70}
}
% legend (top): 3 stage swatches + wait hatch
\draw[productive,fill=spData!24] (0,3.45) rectangle (0.40,3.70);
\node[anchor=west] at (0.45,3.575) {data};
\draw[productive,fill=spForward!24] (1.20,3.45) rectangle (1.60,3.70);
\node[anchor=west] at (1.65,3.575) {fwd};
\draw[productive,fill=spBackward!24] (2.30,3.45) rectangle (2.70,3.70);
\node[anchor=west] at (2.75,3.575) {bwd};
\draw[waitspan] (3.55,3.45) rectangle (3.95,3.70);
\node[anchor=west] at (4.00,3.575) {host-visible sync wait};
% rank labels
\foreach \y/\rank in {2.4/$r_0$,1.35/$r_1$,0.3/$r_2$} {
  \node[anchor=east] at (-0.2,\y) {\rank};
}
% rank r0 (slow data tail dominates)
\draw[productive,fill=spData!24] (0,2.15) rectangle (6.0,2.65) node[midway] {data tail};
\draw[productive,fill=spForward!24] (6.0,2.15) rectangle (7.0,2.65) node[midway] {fwd};
\draw[productive,fill=spBackward!24] (7.0,2.15) rectangle (8.2,2.65) node[midway] {bwd};
% rank r1
\draw[productive,fill=spData!24] (0,1.10) rectangle (1.0,1.60) node[midway] {data};
\draw[productive,fill=spForward!24] (1.0,1.10) rectangle (2.0,1.60) node[midway] {fwd};
\draw[waitspan] (2.0,1.10) rectangle (7.0,1.60) node[midway,fill=white,inner sep=1pt] {wait in bwd};
\draw[productive,fill=spBackward!24] (7.0,1.10) rectangle (8.2,1.60) node[midway] {bwd};
% rank r2
\draw[productive,fill=spData!24] (0,0.05) rectangle (1.1,0.55) node[midway] {data};
\draw[productive,fill=spForward!24] (1.1,0.05) rectangle (2.1,0.55) node[midway] {fwd};
\draw[waitspan] (2.1,0.05) rectangle (7.0,0.55) node[midway,fill=white,inner sep=1pt] {wait in bwd};
\draw[productive,fill=spBackward!24] (7.0,0.05) rectangle (8.1,0.55) node[midway] {bwd};
% frontier boundaries
\draw[markline] (6.0,-0.55) -- (6.0,2.95);
\node[anchor=north east,align=right,fill=white,inner sep=1pt] at (5.95,-0.55)
  {frontier after data\\$F_{t,1}=6.0$};
\draw[markline] (8.2,-0.55) -- (8.2,2.95);
\node[anchor=north west,align=left,fill=white,inner sep=1pt] at (8.25,-0.55)
  {exposed step time\\$F_{t,S}=8.2$};
% time axis
\draw[->,black!65] (0,-0.25) -- (8.65,-0.25) node[right] {rank-local elapsed time};
% callout
\node[anchor=west,text=black!80,align=left] at (8.55,2.40)
  {naive scalar timer:\\charges $\approx\!5$\,s wait\\to \emph{bwd} on $r_1,r_2$};
\draw[callout] (8.50,2.30) to[bend left=8] (4.6,1.62);
% attribution comparison strip
\node[anchor=west,align=left] at (-0.4,-2.10)
  {\textbf{per-stage max}: $6.0+1.0+6.2=\mathbf{13.2}\,$s\quad\emph{(upstream data max + downstream wait duplicate)}};
\node[anchor=west,align=left] at (-0.4,-2.75)
  {\textbf{frontier}: $a_{\mathrm{data}}+a_{\mathrm{fwd}}+a_{\mathrm{bwd}}=6.0+1.0+1.2=\mathbf{8.2}\,$s\quad\emph{(matches exposed makespan)}};
\end{tikzpicture}
\caption{Synchronization displacement in one logical step. Per-rank scalar timers duplicate the same 5\,s group wait across the waiting ranks' backward spans. A per-stage maximum then selects one downstream backward duplicate and adds it to the upstream data maximum, inflating the attributed total to $13.2$\,s; the frontier identity charges the first exposed advance to the upstream data boundary and recovers the actual exposed makespan, $8.2$\,s. Stage abbreviations \emph{data}, \emph{fwd}, \emph{bwd} denote the data-loading, forward-pass, and backward-pass stages (formal names in Section~\ref{sec:method}).}
\label{fig:sync-displacement}
\end{figure}

The production setting drives every design choice:
\begin{itemize}
  \item \emph{Continuous measurement} with no per-step barriers and no explicit CUDA synchronization on the hot path.
  \item \emph{Stable broad stage semantics} comparable across models, with deeper probes as side evidence only.
  \item \emph{Rank-safe aggregation}: any monitoring collective is opt-in and may fail without failing training.
  \item \emph{Prefetch-aware alignment} that charges a data wait to the step that consumes the batch rather than the loop iteration that called \texttt{next}.
  \item \emph{Evidence-scoped output} that separates accounting, model-scoped attribution, and telemetry quality so automation does not add unsupported assumptions.
\end{itemize}
Stage taxonomy and ordered-stage contract are in Appendix~\ref{app:taxonomy}. We use three terms precisely: \emph{low-volume} is a structural property of the telemetry, \emph{low-overhead} is a measured property of an implementation, and \emph{always-on} additionally requires bounded queues, symmetric failure-safe collection, and conservative downgrades.

\section{Frontier Accounting}
\label{sec:method}

\sys{} records a fixed list of non-overlapping trainer-stage spans per rank using CPU wall-clock time. The six default stages are \texttt{data.next\_wait}, \texttt{model.fwd\_loss\_cpu\_wall}, \texttt{model.backward\_cpu\_wall}, \texttt{callbacks.cpu\_wall}, \texttt{optim.step\_cpu\_wall}, and a residual \texttt{step.other\_cpu\_wall} that absorbs closure error (Appendix~\ref{app:taxonomy}). For step $t$, rank $r$, and ordered frontier stage $s$, let $d_{t,r,s}\ge 0$ be the measured duration. Define the rank-local prefix and the max-prefix frontier:
\begin{equation}
  P_{t,r,s}=\sum_{j\le s}d_{t,r,j},\qquad
  F_{t,s}=\max_{r\in\ranks}P_{t,r,s},\qquad
  a_{t,s}=F_{t,s}-F_{t,s-1}.
\end{equation}
Set $F_{t,0}=0$, so $a_{t,1}=F_{t,1}$. Because rank-local prefixes are nondecreasing, the frontier is nondecreasing and $a_{t,s}\ge 0$. For a window $\steps$ of $N=\lvert\steps\rvert$ steps, the reported stage share is step-time weighted:
\begin{equation}
  A_s = \frac{\sum_{t\in\steps} a_{t,s}}{\sum_{t\in\steps} F_{t,S}}.
\end{equation}
The accounted quantity is the maximum rank-local exposed duration for an aligned logical step, $F_{t,S}$, not a wall-clock span from earliest host start to latest finish; the method therefore needs only step-index agreement and rank-local monotonic durations, not synchronized host clocks. Cross-step prefetch and run-ahead break that alignment; \sys{} then either reassigns \texttt{data.next\_wait} to the consuming step (Appendix~\ref{app:taxonomy}) or, when reassignment is unsafe, emits a telemetry-quality downgrade rather than an attribution. Figure~\ref{fig:frontier-accounting} illustrates the construction where a different rank bounds the frontier at each boundary.

\begin{figure}[!ht]
\centering
\begin{tikzpicture}[x=1.9cm,y=0.57cm,font=\footnotesize,>={Latex[length=4.5pt,width=3.5pt]}]
\tikzset{
  frontier/.style={spCallback!90!black,line width=1.4pt,opacity=0.65},
  r0style/.style={line width=0.55pt,black!88,solid},
  r1style/.style={line width=0.55pt,black!78,densely dashed},
  r2style/.style={line width=0.55pt,black!68,densely dotted},
  advance/.style={<->,draw=spCallback!85!black,line width=0.7pt},
  advanceguide/.style={densely dotted,black!35,line width=0.3pt},
  r0mark/.style={fill=black!90,inner sep=0pt,minimum size=3.2pt,circle},
  r1mark/.style={fill=black!80,inner sep=0pt,minimum size=3.0pt,rectangle},
  r2mark/.style={fill=black!70,inner sep=0pt,minimum size=4.2pt,regular polygon,regular polygon sides=3}
}
\draw[->,black!65] (0,0) -- (3.55,0);
\draw[->,black!65] (0,0) -- (0,9.8) node[above,font=\scriptsize,black!65] {prefix time (s)};
\foreach \y in {2,4,6,8} {
  \draw[black!40,line width=0.3pt] (-0.04,\y) -- (0,\y);
  \node[anchor=east,font=\scriptsize,text=black!65] at (-0.06,\y) {\y};
}
\foreach \x in {1,2,3} {
  \draw[black!18] (\x,0) -- (\x,9.0);
}
\node[anchor=north,text=spData!55!black] at (0.5,-0.10) {data};
\node[anchor=north,text=spForward!45!black] at (1.5,-0.10) {fwd};
\node[anchor=north,text=spBackward!55!black] at (2.5,-0.10) {bwd};
\node[r0mark] at (0.10,9.40) {}; \draw[r0style] (0.18,9.40) -- (0.42,9.40);
\node[anchor=west,font=\scriptsize] at (0.44,9.40) {$\,r_0$};
\node[r1mark] at (1.20,9.40) {}; \draw[r1style] (1.28,9.40) -- (1.52,9.40);
\node[anchor=west,font=\scriptsize] at (1.54,9.40) {$\,r_1$};
\node[r2mark] at (2.30,9.40) {}; \draw[r2style] (2.38,9.40) -- (2.62,9.40);
\node[anchor=west,font=\scriptsize] at (2.64,9.40) {$\,r_2$};
\draw[frontier] (3.40,9.40) -- (3.70,9.40);
\node[anchor=west,font=\scriptsize] at (3.72,9.40) {$\,$frontier $F_{t,s}$};
\draw[r0style] (0,0) -- (1,4) -- (2,5) -- (3,8);
\draw[r1style] (0,0) -- (1,1) -- (2,6) -- (3,8);
\draw[r2style] (0,0) -- (1,1) -- (2,3) -- (3,8.5);
\draw[frontier] (0,0) -- (1,4) -- (2,6) -- (3,8.5);
\draw[spCallback!90!black,line width=1.0pt] (1,4)   circle (2.7pt);
\draw[spCallback!90!black,line width=1.0pt] (2,6)   circle (2.7pt);
\draw[spCallback!90!black,line width=1.0pt] (3,8.5) circle (2.7pt);
\node[r0mark] at (1,4) {}; \node[r0mark] at (2,5) {}; \node[r0mark] at (3,8) {};
\node[r1mark] at (1,1) {}; \node[r1mark] at (2,6) {}; \node[r1mark] at (3,8) {};
\node[r2mark] at (1,1) {}; \node[r2mark] at (2,3) {}; \node[r2mark] at (3,8.5) {};
\foreach \y in {0.15,4,6,8.5} {
  \draw[advanceguide] (3.0,\y) -- (3.45,\y);
}
\draw[advance] (3.45,0.15) -- (3.45,4);   \node[anchor=west,text=spCallback!85!black] at (3.55,2.1)  {$a_{\mathrm{data}}=4.0\,$s (from $r_0$)};
\draw[advance] (3.45,4)    -- (3.45,6);   \node[anchor=west,text=spCallback!85!black] at (3.55,5.0)  {$a_{\mathrm{fwd}}=2.0\,$s (from $r_1$)};
\draw[advance] (3.45,6)    -- (3.45,8.5); \node[anchor=west,text=spCallback!85!black] at (3.55,7.25) {$a_{\mathrm{bwd}}=2.5\,$s (from $r_2$)};
\node[anchor=west,text=spCallback!85!black] at (3.55,1.00) {$\sum_s a_s = F_{t,S} = \mathbf{8.5}\,$s};
\end{tikzpicture}
\caption{Max-prefix frontier accounting on a three-rank scenario where a different rank bounds the frontier at each boundary. Each rank's prefix curve $P_{t,r,s}$ is drawn at every boundary; the frontier $F_{t,s}=\max_r P_{t,r,s}$ (purple) takes the latest rank: $r_0$ at data, $r_1$ at fwd, $r_2$ at bwd. The vertical advances form the additive decomposition $F_{t,S}=4.0+2.0+2.5=8.5$\,s, attributing each increment exactly once to the rank attaining the frontier. This is an accounting identity; reading it as a cause requires the model of Section~\ref{sec:sync-wait}.}
\label{fig:frontier-accounting}
\end{figure}
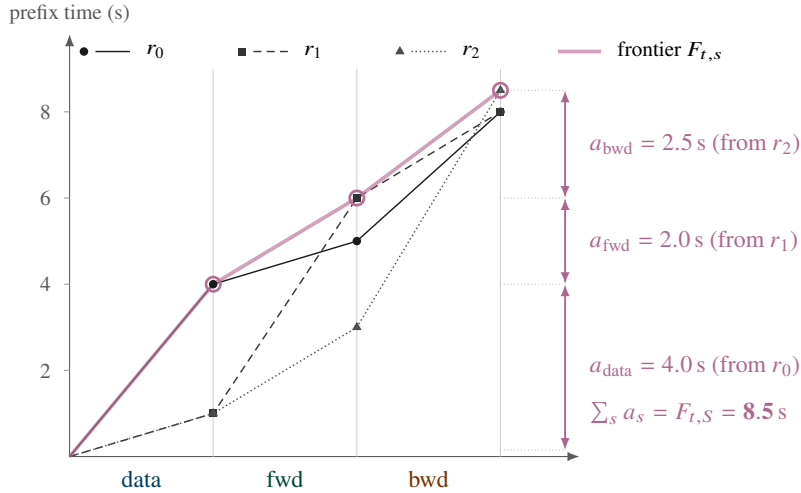

The accounting core is a single pass over steps, ranks, and stages:
\begin{verbatim}
for t in window:
    F_prev = 0
    for s in ordered_stages:
        F = max over r of prefix_sum(d[t,r,1:s])
        a[t,s] = F - F_prev; F_prev = F
\end{verbatim}
This performs $O(RNS)$ arithmetic per window and streams one step at a time in $O(RS)$ memory once that step's rank-stage matrix is available. Edge cases are explicit: with $R=1$ the frontier reduces to the rank's ordered vector and gives no cross-rank evidence; below a small window-denominator floor the implementation emits raw advances rather than percentages; missing stages, mismatched schemas, or mixed world sizes close the current window.

\begin{theorem}[Additive exposed-makespan accounting]
\label{thm:additive}
For any aligned nonnegative stage vectors, $\sum_{s=1}^{S} a_{t,s}=F_{t,S}$.
\end{theorem}
\begin{proof}
Telescoping: $\sum_s(F_{t,s}-F_{t,s-1})=F_{t,S}-F_{t,0}=F_{t,S}$.
\end{proof}

\paragraph{Slack identity.}
Define rank-$r$ slack at boundary $s$ as $\lambda_{t,r,s}=F_{t,s-1}-P_{t,r,s-1}\ge 0$. Then
\begin{equation}
  a_{t,s}=\max_r\bigl(d_{t,r,s}-\lambda_{t,r,s}\bigr),
\end{equation}
so a rank that arrives early at $s-1$ has its current-stage duration discounted by exactly the prior slack it owes the group. This is why a slow data step that forces others to wait is charged once to the data boundary, not again to backward on the ranks that waited.

\paragraph{Comparison with per-stage max and average.}
Let $M_t=\sum_s\max_r d_{t,r,s}$ and $\overline M_t=\sum_s R^{-1}\sum_r d_{t,r,s}$.
\begin{proposition}[Per-stage max overcounting]
\label{prop:max-overcount}
$F_{t,S}\le M_t\le \min(R,S)\,F_{t,S}$; the upper bound is tight.
\end{proposition}
\begin{proposition}[Per-stage average undercounting]
\label{prop:avg-underreport}
$R^{-1}\,F_{t,S}\le \overline M_t\le F_{t,S}$; the lower bound is tight.
\end{proposition}
Per-stage maxima can charge the same exposed second up to $\min(R,S)$ times; per-stage averages can hide rank tails by a factor of $R$. Proofs and a companion measurement-error stability result (errors of order $s\epsilon$ on prefixes and $(2s-1)\epsilon$ on advances) are in Appendix~\ref{app:frontier-bounds-proofs}. For accumulation factor $m$, the ordered list is expanded by accumulation index before the frontier is taken, and semantic reporting groups are aggregated only afterward so repeated microsteps are not collapsed prematurely; changed factors or sync patterns close the window.

\section{Synchronization-Wait Model and Evidence}
\label{sec:labels}
\label{sec:sync-wait}

The frontier is always an accounting signal; bottleneck routing requires an explicit synchronization-wait exposure model, since the same scalar matrix can represent sync wait or independent overlapped work. Write $d_{t,r,s}=x_{t,r,s}+q_{t,r,s}$, with $x$ productive rank-local work and $q$ host-visible waiting charged to the stage in which the host observes it (a kernel launched in forward but observed in backward is recorded in backward; a collective wait is recorded in its enclosing trainer stage). Profilers and wait-state analyzers observe parts of $q$ directly --- as idle, synchronization, or communication events with dependency edges --- and attribute them. \sys{} instead treats $q$ as latent, reasoning only from the ordered progress matrix $d$ in which $q$ is lumped with productive work $x$. It recovers the part of $q$ that is \emph{exposed} in the makespan without ever separating $q$ from $x$ or observing it as a typed event.

A wait segment $q_{t,r,s}$ has a dependency edge to another rank's completion of an earlier ordered boundary and shrinks only when that predecessor prefix advances. Under the assumptions of Table~\ref{tab:sync-assumptions} --- residual closure, aligned steps, common boundary semantics, monotonic timing, wait-dependency charging, homogeneous productive work (or role-aware grouping), and a sufficient dominance margin --- the slack identity implies that the only newly exposed group time at a boundary is the amount by which some rank's prefix exceeds the previous frontier, so frontier advances identify the earliest measured boundary at which delay becomes exposed. The table also gives the evidence label emitted when an assumption is doubtful.

\begin{table}[t]
\centering
\caption{Assumptions of the synchronization-wait exposure model and the evidence label emitted when one is doubtful (deterministic gates in Appendix~\ref{app:labeler-gates}).}
\label{tab:sync-assumptions}
\small
\begin{tabularx}{\linewidth}{@{}YL{0.36\linewidth}@{}}
\toprule
Assumption & If doubtful \\
\midrule
Residual closure and aligned logical steps & \texttt{telemetry\_limited} \\
Common ordered boundary semantics & \texttt{role\_aware\_needed} \\
Rank-local monotonic timing & \texttt{telemetry\_limited} \\
Wait-dependency charging & \texttt{co\_critical} \\
Homogeneous productive work or role-aware grouping & \texttt{role\_aware\_needed} \\
Sufficient dominance margin & \texttt{co\_critical} \\
\bottomrule
\end{tabularx}
\end{table}

The causal reading --- that this time was \emph{caused} by an upstream straggler rather than independent overlapped work --- follows from wait-dependency charging plus rank comparability; when either is questionable the labeler emits \texttt{co\_critical} or \texttt{role\_aware\_needed}. As a sharp case, the two-rank matrix $r_0=(10,0)$, $r_1=(0,10)$ gives $F_{\mathrm{data}}=F_{\mathrm{bwd}}=10$ and charges $10$ to data and $0$ to backward, but is equally consistent with $r_1$ waiting on $r_0$'s data (data causal) or with two independent co-critical paths; \sys{} emits the model-scoped attribution only when workload semantics support it, and otherwise carries a \texttt{co\_critical} ambiguity set such as \{\texttt{data}, \texttt{backward}\}, reported as the machine-readable \texttt{co\_critical\_stages} field.

A complementary direct-exposure score replaces stage $s$ with a clipped baseline and recomputes the frontier. Let $\tilde b_{t,r,s}$ be a candidate baseline for stage $s$ --- a per-rank window median, cohort median, or no-stall reference --- clipped to $b_{t,r,s}=\min(d_{t,r,s},\tilde b_{t,r,s})$ so the replacement never exceeds the observation:
\begin{equation}
  G_s(b)=\frac{\sum_t \bigl(F_{t,S}-F^{(s\leftarrow b)}_{t,S}\bigr)}{\sum_t F_{t,S}}\ge 0.
\end{equation}
For a feasible baseline whose stage-$s$ reduction also removes the downstream wait it induces, $G_s$ lower-bounds the model-scoped gain; otherwise it is a conservative sensitivity score, not an intervention estimate, since the recomputation leaves any non-removable downstream wait in place. High $A_s$ with high $G_s$ supports \texttt{direct\_exposure}; high $A_s$ with low $G_s$ supports \texttt{sync\_wait\_dependent} when workload semantics agree and \texttt{co\_critical} otherwise.

For localization the labeler reports the latest-rank tie set, the lag $L_{t,s}=\max_r P_{t,r,s}-\median_r P_{t,r,s}$ and its increment, and the max-minus-secondmax gap, counting switches only between confident unique leaders. Labels (Table~\ref{tab:labels-core}) describe orthogonal evidence axes, not a flat confidence ladder; the full label set and deterministic downgrade gates are in Appendices~\ref{app:full-labels}--\ref{app:labeler-gates}. The \emph{routing candidate set} $C_{\mathrm{route}}$ is the smallest leading-share prefix whose cumulative share reaches $\tau_C$ (default $0.80$); the evaluation reports \emph{top-2} (seeded stage among the two highest shares) and \emph{candidate hit} (anywhere in the prefix), always paired with candidate-set size. The routing set is kept separate from the ambiguity set (\texttt{co\_critical}), so a row can route to a compact set while still flagging the stages that remain jointly plausible.

\begin{table}[htbp]
\centering
\caption{Core diagnosis labels emitted in the main pipeline.}
\label{tab:labels-core}
\small
\begin{tabularx}{\linewidth}{@{}L{0.27\linewidth}Y@{}}
\toprule
\textbf{Label} & \textbf{Meaning} \\
\midrule
\texttt{frontier\_accounting} & Additive exposed-makespan decomposition; base claim. \\
\texttt{direct\_exposure} & Raw duration, spread, and clipped static gain agree with the frontier stage. \\
\texttt{sync\_wait\_dependent} & Frontier share high but static gain low; actionability depends on the wait model. \\
\texttt{co\_critical} & Multiple stages or ranks can plausibly remain bottlenecks after optimizing one stage. \\
\texttt{role\_aware\_needed} & Rank roles differ; global rank aggregation is unsafe. \\
\texttt{telemetry\_limited} & Residuals, gather failures, or missing probes cap confidence. \\
\bottomrule
\end{tabularx}
\end{table}

\section{PyTorch Implementation}
\label{sec:impl}

The minimal integration wraps the logical optimizer step with one step context and ordered, non-overlapping stage contexts:
\begin{verbatim}
with perf.step():
    with perf.stage("data.next_wait"): batch = next(data_iter)
    with perf.stage("model.fwd_loss_cpu_wall"):
        with perf.cuda_event_sample("model.fwd_loss_cuda_event_ms"):
            loss = forward_loss(batch)
    with perf.stage("model.backward_cpu_wall"): loss.backward()
    with perf.stage("callbacks.cpu_wall"): run_callbacks()
    with perf.stage("optim.step_cpu_wall"): optimizer.step()
\end{verbatim}
For DDP --- and for the PyTorch \texttt{ZeroRedundancyOptimizer} spot check (ZeRO-1-style optimizer-state sharding, distinct from DeepSpeed ZeRO) --- this placement includes reducer activity and exposed collective waits in the backward stage; for FSDP \texttt{FULL\_SHARD}, all-gather, reduce-scatter, and parameter materialization are charged to the broad forward, backward, or optimizer region in which the host observes them.

\sys{} avoids \texttt{torch.cuda.synchronize()} on the hot path so monitoring does not perturb overlap; the price is that stage timing follows host-visible exposure, which is exactly what the frontier identity decomposes, while kernel-level attribution remains the profiler's job. An optional CUDA-event forward channel, sampled at deterministic fraction $q\in\{0,0.05,1\}$, records two timing events around the forward/loss region and polls readiness at later safe points~\cite{pytorchcudaevent}; the event value is side evidence only, never enters the prefix vector, and can support a \texttt{forward\_device\_supported} label or suggest host overhead when CPU-wall forward is high but event time is low. A window-boundary collective gathers the per-rank $[N,S]$ buffer to rank 0 over a separate Gloo or NCCL telemetry process group; a failed or timed-out gather records \texttt{gather\_ok=false}, emits any safe local summary, and downgrades distributed labels to \texttt{telemetry\_limited}, never failing training. For rank count $R$, $N$ steps, $K$ ordered stage fields, and $b$ bytes, the dense root-visible payload is $B_{\mathrm{root}}=RNKb$ ($\approx 0.61$ MB at $R{=}128,N{=}100,K{=}6,b{=}8$; the 32-rank E9 evidence packets of Section~\ref{sec:eval-profiler-comparison} measure $0.11$ MB); fractional overhead is the gather-path time divided by training time in the window, so longer windows amortize it.

\section{Evaluation}
\label{sec:eval}

We evaluate \sys{} with the \texttt{stagefrontier-artifact} repository. Cluster experiments use bf16 transformer training in the NVIDIA PyTorch 24.12 container; DDP is the main validation path, Gloo telemetry gather is the main backend, and the overhead matrix additionally evaluates NCCL. The main E0--E5 routing/overhead matrix disables gradient accumulation; E6 is a removed-injection A/B/A check, E7 a fixed-factor gradient-accumulation spot check, E8 a scope-extension spot check for PyTorch FSDP \texttt{FULL\_SHARD} and \texttt{ZeroRedundancyOptimizer} ZeRO-1, and E9 a 32-rank selected-window comparison against PyTorch Profiler, HTA, and Nsight Systems. Faults are host-side delays injected into the CPU-wall span of one rank (12--360\,ms depending on scenario); most rows surface the delay as host-visible time, and only the synchronization-bearing rows additionally place a \texttt{torch.distributed.barrier()} in the affected span so the delay is exposed rather than absorbed. Windows span 120--600 measured steps after 20 warmup steps, the fractional overhead $\rho$ divides gather-path time by training time in the window, and the paired bootstrap resamples whole run/window blocks rather than individual steps. Experiment groups are defined in Appendix~\ref{app:artifact-checklist} (Table~\ref{tab:experiment-groups}); per-group model and hardware configuration, exact step counts, and bootstrap scripts are recorded in the artifact's structured rows and Slurm wrappers. Table~\ref{tab:rq-evidence-gap} maps each research question to its supporting evidence.

The evaluation addresses five research questions:
\begin{description}[leftmargin=2.2em,labelwidth=2.0em,topsep=3pt,itemsep=1pt]
  \item[RQ1] Does \sys{} handle synchronization displacement better than existing rank-local stage summaries?
  \item[RQ2] Can it route bottlenecks to a compact candidate stage and rank \emph{without} a full profiler?
  \item[RQ3] Does it remain valid under modern distributed training regimes such as gradient accumulation and FSDP/ZeRO sharding?
  \item[RQ4] Is the overhead low enough for always-on deployment?
  \item[RQ5] Does the routing signal translate to operational value on real workloads?
\end{description}

\begin{table}[t]
\centering
\caption{Evidence and supporting sections for each research question (RQ1--RQ5).}
\label{tab:rq-evidence-gap}
\small
\begin{tabularx}{\linewidth}{@{}L{0.05\linewidth}YL{0.18\linewidth}@{}}
\toprule
\textbf{RQ} & \textbf{Evidence} & \textbf{Section} \\
\midrule
RQ1 & Telescoping identity at floating-point roundoff; max/avg bounds; 100\% sync-wait fixture recovery vs.\ 0\% for max/average. & \ref{sec:eval-validation}, \ref{sec:eval-routing} \\
RQ2 & Five fault classes 50/50 top-2 (40/50 top-1) at 8/32 ranks plus 64/128-rank spot checks; event channel separates forward-device vs.\ host; 32-rank selected windows agree with PyTorch Profiler / HTA / Nsight on 12/12 positive rows under a shared reducer. & \ref{sec:eval-routing}, \ref{sec:eval-profiler-comparison} \\
RQ3 & Fixed-factor grad-accum spot check passes; FSDP \texttt{FULL\_SHARD} and ZeRO-1 route 90/90 sync-bounded rows top-2, 87/90 top-1. & \ref{sec:scoped-checks} \\
RQ4 & 95\% CI upper bound $\le 0.181\%$ CPU-wall and $\le 0.043\%$ event-channel through 128 ranks on Gloo and NCCL. & \ref{sec:overhead} \\
RQ5 & Five field deployments routed investigators to dataloader, callback, post-op, and QAT actions; two corroborated by a profiler trace and an op-level ablation. & \ref{sec:cases} \\
\bottomrule
\end{tabularx}
\end{table}

\subsection{Algorithmic validation (RQ1)}
\label{sec:eval-validation}
The artifact CPU validation suite covers properties independent of the cluster campaign. The telescoping identity holds to floating-point roundoff ($8.88\times 10^{-16}$ max error); the formal max/average bounds of Propositions~\ref{prop:max-overcount}--\ref{prop:avg-underreport} are satisfied on random and tight fixtures (0 violations); measurement-error stability holds (observed/bound $\le 0.9998$). On the sync-wait fixture ($n=120$), \sys{} recovers the upstream boundary in 100\% of rows while per-stage max and average recover it in 0\%; direct-exposure recovery is 100\% ($n=240$); and four downgrade fixtures (co-critical, role-heterogeneous, telemetry-limited, two-stage tied) all trigger their expected labels.

\subsection{Hidden-rank candidate routing and baselines (RQ1/RQ2)}
\label{sec:eval-routing}
E3 injects 120 ms delays into one rank in five hidden-rank scenarios (data, backward, backward/comm, forward/device, forward/host) at 8 and 32 ranks (5 seeds each), plus a callback-sync pilot. Figure~\ref{fig:eval-results} summarizes the two headline results --- compact-routing counts against baselines and the data-tail detectability transition --- and Table~\ref{tab:routing-and-baselines} gives the exact counts.

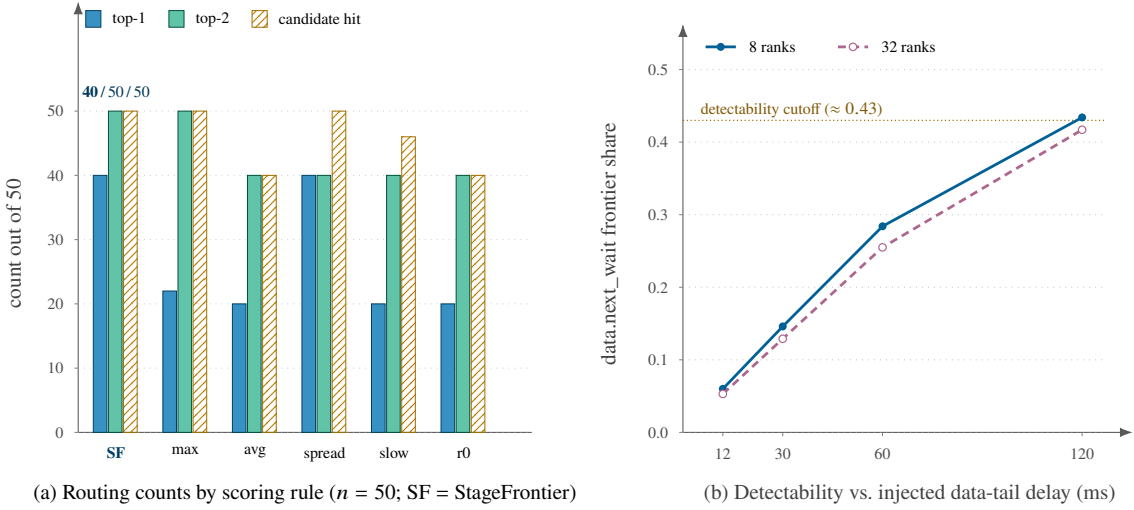
\begin{figure}[t]
\centering
\begin{tikzpicture}[font=\scriptsize,>=Latex]
% ==================== Panel (a): Routing comparison ====================
\begin{scope}[xshift=0cm,yshift=0cm]
\tikzset{
  bartop1/.style={fill=spData!78,draw=spData!55!black,line width=0.25pt},
  bartop2/.style={fill=spForward!62,draw=spForward!50!black,line width=0.25pt},
  barcand/.style={fill=spWait!30,draw=spWait!70!black,line width=0.35pt,pattern=north east lines,pattern color=spWait!80!black}
}
\draw[->,black!65] (0,0) -- (0,5.70);
\node[anchor=south,font=\scriptsize,rotate=90,text=black!75] at (-0.62,2.5) {count out of 50};
\draw[black!65] (0,0) -- (6.0,0);
\foreach \y in {0,10,20,30,40,50} {
  \pgfmathsetmacro{\yp}{\y*0.085}
  \draw[black!30,line width=0.25pt,dotted] (0,\yp) -- (5.95,\yp);
  \draw[black!50,line width=0.3pt] (-0.05,\yp) -- (0,\yp);
  \node[anchor=east,font=\tiny,text=black!70] at (-0.06,\yp) {\y};
}
\foreach \i/\name/\a/\b/\c in {%
  0/SF/40/50/50,%
  1/max/22/50/50,%
  2/avg/20/40/40,%
  3/spread/40/40/50,%
  4/slow/20/40/46,%
  5/r0/20/40/40%
} {
  \pgfmathsetmacro{\gx}{0.50 + \i*0.92}
  \pgfmathsetmacro{\ya}{\a*0.085}
  \pgfmathsetmacro{\yb}{\b*0.085}
  \pgfmathsetmacro{\yc}{\c*0.085}
  \draw[bartop1] ({\gx-0.30},0) rectangle ({\gx-0.12},\ya);
  \draw[bartop2] ({\gx-0.10},0) rectangle ({\gx+0.08},\yb);
  \draw[barcand] ({\gx+0.10},0) rectangle ({\gx+0.28},\yc);
}
\node[anchor=north,font=\tiny\bfseries,text=spData!55!black] at ({0.50+0*0.92},-0.06) {SF};
\foreach \i/\name in {1/max,2/avg,3/spread,4/slow,5/r0} {
  \node[anchor=north,font=\tiny] at ({0.50+\i*0.92},-0.06) {\name};
}
\node[anchor=south,font=\tiny,text=spData!50!black] at ({0.50},{50*0.085+0.05}) {\textbf{40}\,/\,50\,/\,50};
\draw[bartop1] (0.10,5.36) rectangle (0.30,5.58);
\node[anchor=west,font=\tiny] at (0.32,5.47) {top-1};
\draw[bartop2] (1.20,5.36) rectangle (1.40,5.58);
\node[anchor=west,font=\tiny] at (1.42,5.47) {top-2};
\draw[barcand] (2.30,5.36) rectangle (2.50,5.58);
\node[anchor=west,font=\tiny] at (2.52,5.47) {candidate hit};
\node[anchor=north,font=\scriptsize] at (3.0,-0.55) {(a) Routing counts by scoring rule ($n=50$; SF $=$ \sys{})};
\end{scope}
% ==================== Panel (b): Detectability curve ====================
\begin{scope}[xshift=8.0cm,yshift=0cm]
\def\sxb{0.044}
\def\syb{9.6}
\draw[->,black!65] (0,0) -- (0,5.40);
\node[anchor=south,font=\scriptsize,rotate=90,text=black!75] at (-0.72,2.5) {data.next\_wait frontier share};
\draw[->,black!65] (0,0) -- (5.95,0);
\node[anchor=north,font=\scriptsize,text=black!75] at (3.0,-0.55) {(b) Detectability vs.\ injected data-tail delay (ms)};
\foreach \y in {0.0,0.1,0.2,0.3,0.4,0.5} {
  \pgfmathsetmacro{\yp}{\y*\syb}
  \draw[black!35,line width=0.25pt,dotted] (0,\yp) -- (5.6,\yp);
  \draw[black!50,line width=0.3pt] (-0.05,\yp) -- (0,\yp);
  \node[anchor=east,font=\tiny,text=black!70] at (-0.06,\yp) {\y};
}
\foreach \x in {12,30,60,120} {
  \pgfmathsetmacro{\xp}{\x*\sxb}
  \draw[black!50,line width=0.3pt] (\xp,-0.05) -- (\xp,0);
  \node[anchor=north,font=\tiny,text=black!70] at (\xp,-0.07) {\x};
}
\draw[densely dotted,spWait!70!black,line width=0.45pt] (0,0.43*\syb) -- (5.6,0.43*\syb);
\node[anchor=west,font=\tiny,text=spWait!55!black] at (0.10,0.43*\syb+0.14) {detectability cutoff ($\approx 0.43$)};
\draw[line width=1.05pt,spData!85!black]
  (12*\sxb,0.060*\syb) -- (30*\sxb,0.146*\syb) -- (60*\sxb,0.284*\syb) -- (120*\sxb,0.434*\syb);
\foreach \pt in {(12*\sxb,0.060*\syb),(30*\sxb,0.146*\syb),(60*\sxb,0.284*\syb),(120*\sxb,0.434*\syb)} {
  \fill[spData!85!black] \pt circle (1.5pt);
}
\draw[line width=1.05pt,spCallback!85!black,densely dashed]
  (12*\sxb,0.053*\syb) -- (30*\sxb,0.129*\syb) -- (60*\sxb,0.255*\syb) -- (120*\sxb,0.417*\syb);
\foreach \pt in {(12*\sxb,0.053*\syb),(30*\sxb,0.129*\syb),(60*\sxb,0.255*\syb),(120*\sxb,0.417*\syb)} {
  \draw[draw=spCallback!85!black,fill=white,line width=0.45pt] \pt circle (1.4pt);
}
\draw[line width=1.05pt,spData!85!black] (0.30,5.10) -- (0.70,5.10);
\fill[spData!85!black] (0.50,5.10) circle (1.5pt);
\node[anchor=west,font=\tiny] at (0.74,5.10) {8 ranks};
\draw[line width=1.05pt,spCallback!85!black,densely dashed] (2.05,5.10) -- (2.45,5.10);
\draw[draw=spCallback!85!black,fill=white,line width=0.45pt] (2.25,5.10) circle (1.4pt);
\node[anchor=west,font=\tiny] at (2.49,5.10) {32 ranks};
\end{scope}
\end{tikzpicture}
\caption{Headline empirical results for the hidden-rank routing matrix. \textbf{(a)} Routing counts across five E3 injection scenarios ($n{=}50$); candidate-hit counts should be read with the candidate-set sizes in Table~\ref{tab:routing-and-baselines}. \textbf{(b)} Mean \texttt{data.next\_wait} frontier share vs.\ injected data-tail delay shows a detectability transition: the dotted line marks the empirical single-stage data share ($\approx 0.43$, at 120\,ms) above which \texttt{data.next\_wait} enters the compact $\tau_C{=}0.80$ candidate prefix.}
\label{fig:eval-results}
\end{figure}

Each baseline applies one stage-attribution rule to the same $[N,R,S]$ window matrix used by \sys{}, sharing windowing, schema validation, and tie tolerance, so the counts isolate the scoring rule. \emph{Per-stage max} and \emph{average} rank stages by their max/mean share; \emph{raw rank spread} ranks by $\sum_t(\max_r d_{t,r,s}-\median_r d_{t,r,s})$, a dispersion heuristic with no stage-attribution semantics; \emph{slowest-rank breakdown} reports the largest stage of the per-step slowest rank; \emph{rank-0 local total} ignores all other ranks.

\begin{table}[t]
\centering
\caption{Routing on E3 120\,ms DDP injection rows (5 scenarios $\times$ 2 rank counts $\times$ 5 seeds = 50). Candidate-hit counts are paired with average/maximum candidate-set size. \sys{} is the only method that both exactly decomposes exposed makespan and keeps the candidate set size two on every row. Forward/device CPU-wall routing and event side evidence are separated in Table~\ref{tab:forward-device-split}.}
\label{tab:routing-and-baselines}
\scriptsize
\begin{tabular}{@{}L{0.27\linewidth}crrrrr@{}}
\toprule
Method & Exact acct. & Top-1 & Top-2 & Cand.\ hit & Avg cand. & Max cand. \\
\midrule
\sys{}                 & yes & 40/50 & 50/50 & 50/50 & 2.00 & 2 \\
Per-stage max          & no  & 22/50 & 50/50 & 50/50 & 2.40 & 3 \\
Per-stage average      & no  & 20/50 & 40/50 & 40/50 & 2.30 & 3 \\
Raw rank spread        & no  & 40/50 & 40/50 & 50/50 & 2.20 & 3 \\
Slowest-rank breakdown & no  & 20/50 & 40/50 & 46/50 & 2.40 & 3 \\
Rank-0 local total     & no  & 20/50 & 40/50 & 40/50 & 2.16 & 3 \\
\bottomrule
\end{tabular}
\end{table}

\sys{} routes 50/50 rows into the top-2 candidate set and 40/50 into top-1. Per-stage max ties on top-2 but misses 28 top-1 calls because downstream exposed waits dominate the raw maximum, and the same summary overcounts $F_{t,S}$ by up to $\min(R,S)$ (Proposition~\ref{prop:max-overcount}); mean, slowest-rank, and rank-0 summaries miss data- and rank-tail cases more often. Raw rank spread ties \sys{} on top-1 but misses 10 top-2 cases, reaches its candidate-hit count only with a larger emitted set, and still has no stage-attribution semantics under the non-identifiability example of Section~\ref{sec:sync-wait}. The headline advantage is therefore the combination of \emph{exact additive accounting}, \emph{compact top-2 candidate routing}, and \emph{labels that scope the interpretation}. The candidate set is size two on all 50 rows (Table~\ref{tab:stagefrontier-final-top12}), and the candidate hit is stable for $\tau_C\in[0.70,0.90]$ (Table~\ref{tab:tau-sensitivity}). Lower-magnitude data tails fall below the compact-routing threshold rather than being misattributed: the mean \texttt{data.next\_wait} share rises $0.06\to0.43$ from 12 to 120 ms at 8 ranks (Figure~\ref{fig:eval-results}b). A profiler-enabled calibration places the cumulative-prefix crossing of $\tau_C{=}0.80$ between the 120 and 180 ms magnitudes ($0.795\to0.823$ at 8 ranks, $0.790\to0.818$ at 32), grounding the detectability threshold in measured shares. Selected 64/128-rank spot checks route communication and 180 ms data-tail injections into the top-2 in all checked seeds, while the 120 ms data row at 128 ranks (Delay/p50 $\approx 0.21$) likewise falls below the threshold rather than misrouting.

Forward/device injections are not claimed as a CPU-wall top-1 attribution: forward CUDA work launched on the main stream commonly becomes host-visible later in backward, so the broad-stage prefix legitimately ranks backward first. The row counts as a successful profiler routing only when the seeded forward stage stays in the compact top-2 \emph{and} the sampled event channel ($q\in\{0.05,1\}$) emits \texttt{forward\_device\_supported} (Table~\ref{tab:forward-device-split}). The synchronization-bearing callback rows show a magnitude boundary: 0/3 top-1 at 60 and 120 ms (top-2 3/3), then 3/3 top-1 at 180 and 240 ms; the magnitude calibration plus the A/B/A check below carry that family's claim.

\begin{table}[t]
\centering
\caption{Forward/device claim separation. CPU-wall frontier accounting supplies compact routing; optional CUDA-event side evidence supplies device support.}
\label{tab:forward-device-split}
\small
\begin{tabularx}{\linewidth}{@{}L{0.20\linewidth}L{0.20\linewidth}L{0.17\linewidth}Y@{}}
\toprule
Fault family & CPU-wall top-1 & CPU-wall top-2 & Event evidence \\
\midrule
Forward/device & not claimed (0/10 E3) & 10/10 E3 & \texttt{forward\_\allowbreak{}device\_\allowbreak{}supported} on event-enabled E5 rows \\
Forward/host & 10/10 E3 & 10/10 E3 & \texttt{forward\_\allowbreak{}host\_\allowbreak{}overhead\_\allowbreak{}suspected} when CPU-wall is high and event time is low \\
\bottomrule
\end{tabularx}
\end{table}

\subsection{Selected-window profiler comparison (RQ2)}
\label{sec:eval-profiler-comparison}
E9 compares \sys{} with PyTorch Profiler, HTA, and Nsight Systems on the same fixed 32-rank DDP window definition, as a router-vs-trace tradeoff study rather than a replacement claim: the heavier tools reproduce \sys{}'s broad-frontier ranking on a selected window, at much higher selected-window cost. Each tool is replayed on the same selected window so routing and trace costs are measured on identical work (three seeds per scenario, 180 ms hidden-rank injections, the inner 20 of 40 captured steps scored). To keep the comparison about the paper's claimed quantity, each heavy-profiler trace is reduced to the same ordered broad-stage matrix --- Kineto \texttt{record\_function} ranges for PyTorch Profiler and HTA, host-side NVTX ranges for Nsight --- and scored with the max-prefix frontier recurrence of Section~\ref{sec:method}; kernel, CUPTI, and NCCL timelines remain complementary origin/overlap evidence. The positive denominator covers four sync-bounded rows (\texttt{data\_tail}, \texttt{comm\_delay}, \texttt{fwd\_cuda\_compute}, \texttt{callback\_sync\_tail}) over three seeds. The \texttt{fwd\_cuda\_compute} row's 180\,ms injection is large enough that broad CPU-wall forward dominates the frontier and is recovered top-1 here, unlike the lower-magnitude 120\,ms E3 forward/device rows, which are not claimed as a CPU-wall top-1 attribution (Table~\ref{tab:forward-device-split}).

\begin{table}[t]
\centering
\caption{E9 selected-window profiler comparison at 32 ranks (three seeds per scenario; inner 20 of 40 captured steps scored). Top-1/top-2 are computed after reducing each trace to the same max-prefix frontier shares; artifact size, step overhead, and postprocessing are medians from no-fault selected-window rows; \sys{} scores inline and runs no trace-postprocessing pass (shown as \emph{none}). HTA consumes the PyTorch Profiler Kineto trace and so repeats its capture/runtime/storage cells. $\dagger$: within run-to-run noise; use the paired E1 bound (Table~\ref{tab:overhead-short}) for continuous \sys{} overhead.}
\label{tab:profiler-head-to-head}
\scriptsize
\setlength{\tabcolsep}{3pt}
\begin{tabular}{@{}lrrrrrrrr@{}}
\toprule
Tool & Pos.\ rows & Top-1 & Top-2 & Cand.\ hit & Cand.\ avg/max & Artifact & No-fault overhead & Postproc. \\
\midrule
\sys{} & 12 & 12/12 & 12/12 & 12/12 & 2.00/2 & 0.11 MB & within noise$^{\dagger}$ & none \\
PyTorch Profiler & 12 & 12/12 & 12/12 & 12/12 & 2.25/3 & 15.81 GB & 36.6\% & 43.5 s \\
HTA on Kineto & 12 & 12/12 & 12/12 & 12/12 & 2.25/3 & 15.81 GB & 36.6\% & 43.5 s \\
Nsight Systems & 12 & 12/12 & 12/12 & 12/12 & 2.00/2 & 1.22 GB & 50.3\% & 2.59 s \\
\bottomrule
\end{tabular}
\end{table}

Under the shared reducer, \sys{}, PyTorch Profiler/HTA, and Nsight all recover the four positive broad stages top-1 and top-2 across all three seeds (Table~\ref{tab:profiler-head-to-head}), and the ratio vector agrees closely: the largest single-stage share difference from \sys{} on any of the six broad stages (worst case over the per-seed share vectors and the four positive scenarios) is 0.039 for PyTorch/HTA and 0.024 for Nsight, below the default share tie tolerance $\eta_A=0.05$ (Table~\ref{tab:labeler-gates}), so the candidate-routing result is unchanged. \sys{} produces this from a median 0.11 MB evidence packet, versus a 15.81 GB Kineto trace (PyTorch Profiler/HTA, raising no-fault step time 36.6\% during capture) and 1.22 GB of Nsight reports (50.3\%). Perfect 12/12 top-1 agreement is partly structural --- each tool's output is collapsed onto the same broad-stage matrix and scored by the same recurrence --- so this tests whether each heavy trace's stage-range evidence is faithful enough to reproduce \sys{}'s exposed-makespan ranking, not whether the tools' native workflows agree end to end; those native analyses (operator/kernel drill-down, temporal/idle/overlap breakdowns) remain complementary. The host-only \texttt{callback\_host\_tail} control is excluded because no tool routes a callback without a synchronization boundary into the top-2 under this reducer: the cost is visible to the traces but not exposed as group delay, so the frontier correctly leaves it unranked.

\subsection{Overhead, removed-injection, and scoped checks (RQ4/RQ3)}
\label{sec:overhead}
E1 measures overhead with paired runs inside the same submitted workload and backend; the resampling unit is the paired run/window block, so the bound respects within-run dependence. Table~\ref{tab:overhead-short} reports 95\% CI upper bounds on throughput overhead through 128 ranks for Gloo and NCCL gather. Every positive upper bound is below the pre-registered 3\% (CPU-wall) and 1\% (event-channel) gates; the maxima are 0.181\% (16-rank NCCL CPU-wall) and 0.043\% (16-rank Gloo event). Across 105 no-fault rows and 315 windows all gathers completed and no row emitted a strong bottleneck label or downgraded to \texttt{telemetry\_limited}; the one-sided 95\% upper bound on the strong-label rate from 0/105 is 2.81\%, a sanity check at this scale rather than a fleet-grade false-alarm budget. Scale beyond 128 ranks, window-length sensitivity, and long-run gather backpressure are left to future evaluation.

\begin{table}[t]
\centering
\caption{E1 overhead scale; 95\% CI upper bounds on throughput overhead (\%). Each rank/backend/mode cell uses five paired submitted runs with 400-step windows; $\le 0$ denotes a bound at or below zero (overhead within run-to-run noise).}
\label{tab:overhead-short}
\small
\begin{tabular}{@{}rrrrr@{}}
\toprule
Ranks & Gloo CPU-wall & Gloo event-inc & NCCL CPU-wall & NCCL event-inc \\
\midrule
8   & 0.019 & 0.004 & 0.098 & $\le 0$ \\
16  & 0.160 & 0.043 & 0.181 & $\le 0$ \\
32  & 0.047 & 0.023 & 0.155 & $\le 0$ \\
64  & $\le 0$ & $\le 0$ & 0.083 & $\le 0$ \\
128 & 0.107 & $\le 0$ & 0.131 & $\le 0$ \\
\bottomrule
\end{tabular}
\end{table}

\paragraph{Removed-injection consistency.}
\label{sec:aba}
E6 runs three A/B/A windows per seed (3 seeds, 8 ranks, 200 steps each): baseline A1, a window B with a 120 ms callback-sync injection, and a removed-injection A2 under the same seed and allocation. Median step time goes $207.81\to294.12\to208.02$ ms and callback frontier share goes $1.75\%\to41.06\%\to1.75\%$, returning step time to within 0.21 ms of baseline (recovery ratio 0.998, well under the 8.63 ms tolerance). The callback is a stable top-2 candidate (0/3 top-1 at this magnitude), so the result reads as removed-injection consistency for a top-2 candidate, not a top-1 intervention attribution.

\paragraph{Gradient-accumulation and sharding scope checks (RQ3).}
\label{sec:scoped-checks}
A fixed-factor gradient-accumulation check (factor four, DDP \texttt{no\_sync}, 8 ranks, 5 seeds, ordered accumulation-indexed substages) routes data and backward top-1/top-2 on all rows; forward/device stays top-2 (co-critical with backward host time) with the event channel still emitting \texttt{forward\_device\_supported}, and ordered-vs-broad throughput ratios fall in $[0.999,1.001]$. A sharding check wraps the same transformer in PyTorch FSDP \texttt{FULL\_SHARD} and \texttt{ZeroRedundancyOptimizer} ZeRO-1 (8/16/32 ranks, seeds 0--2, 180 ms injections): all 90 sync-bounded positive rows route top-2 and 87/90 top-1 (the three misses are 32-rank ZeRO-1 callback/sync rows), while a host-local optimizer control without an adjacent barrier is 0/18 top-1/top-2: work visible to a rank but not exposed as group delay is correctly left unrouted.

\subsection{Operational evidence for monitoring (RQ5)}
\label{sec:cases}
Since March 2026, \sys{} has been enabled by default in a production NVIDIA AV training environment covering multiple model pipelines. Of five anonymized field deployments (Table~\ref{tab:cases-short}), the two strongest were checked against ground truth: QAT-B against a PyTorch Profiler capture and Mixed-A against an operator ablation (both detailed below). The deployments span sync-wait exposure (Vision-A), direct exposure (QAT-A, QAT-B), and mixed or substage-deep cases (Vision-B, Mixed-A). Two also recorded a post-fix throughput improvement (Vision-A $230\to340$ samples/s at 8 nodes, Vision-B $63\to85$ samples/s single-node).

\begin{table}[t]
\centering
\caption{Representative operational diagnoses. Frontier share is exposed-makespan share from max-prefix accounting; the claim level is the strongest evidence the follow-up established, not a throughput attribution.}
\label{tab:cases-short}
\footnotesize
\begin{tabularx}{\linewidth}{@{}L{0.08\linewidth}L{0.18\linewidth}YL{0.18\linewidth}L{0.16\linewidth}@{}}
\toprule
Case & Frontier signal & Follow-up evidence & Operator action & Claim level \\
\midrule
Vision-A & Data-tail from random ranks (small median, large tail) & Dataloader worker-to-main IPC stall & Main-process prefetch & Consistency \\
Vision-B & fwd 44\%, data 27\%, callback 24\%, bwd 5\% & \texttt{pin\_memory} off, worker count, image logging every 10 steps & Pin memory, worker tuning, reduced logging & Consistency \\
Mixed-A & Data 90\% with large delta-lag & Post-op CPU/numpy/I/O paths (op ablation) & Post-op relocation/vectorization & Attribution-confirmed routing \\
QAT-A & fwd 76\%, bwd 15\%, data 5\%, callback 4\% & QAT fake-quant dominates forward (no separate profiler capture) & QAT/forward overhead & Consistency \\
QAT-B & fwd $\approx$67\%, bwd $\approx$22\%, data $\approx$10\% & Profiler: fake-quant op with sync, transfers, many tiny kernels & Forward/QAT overhead & Profiler-corroborated routing \\
\bottomrule
\end{tabularx}
\end{table}

Two cases were checked against ground truth. In \textbf{QAT-B}, because a high CPU-wall forward share can reflect host overhead rather than device compute, the routing call was corroborated with a short PyTorch Profiler capture that placed the cost in genuine forward device work --- the quantization-aware-training fake-quant modules, whose per-element scalar synchronizations and tens of thousands of sub-$50\,\mu$s kernels dominate the trace. The always-on signal selected the stage and the heavy trace identified the operator-level cause (the intended router-to-profiler relationship); the suggested fusion and sync-removal optimizations were measured offline but judged too invasive to ship. In \textbf{Mixed-A}, the broad data $\approx 90\%$ frontier might be read as a dataloader problem, but the ordered sub-stage breakdown showed the workers essentially idle and the cost in a GPU post-processing stage; an op-level ablation found it was almost entirely one data-transform operator, about $27.7$\,s on the profiled step. Disabling that operator cut it to about $1.3$\,s but exposed a second transform of similar magnitude behind it, so the stage cost did not drop --- though the ablation still confirmed which operator the routing pointed to.

Vision-B also shows how a periodic stage cost hides from the median summaries common in throughput dashboards (Table~\ref{tab:median-vs-frontier}). The callback stage had a per-step median of about $2$\,ms --- roughly $0.13\%$ of the $1.526$\,s median step --- because its cost was concentrated in a spike that recurred every ten steps and reached a $7.17$\,s maximum, so a median dashboard reports it as negligible. Frontier accounting sums the per-step exposed advances and assigns the callback a $24\%$ exposed-makespan share, with the straggler summary flagging a single intermittent rank, and routes the investigator to it; a max or p95 view would reveal the magnitude but not its share of exposed step time. Configuration inspection confirmed the cause: an image-logging callback firing every ten steps.

\begin{table}[t]
\centering
\caption{Vision-B callback stage under three views on the same field window. The median view, common in throughput dashboards, hides the periodic spike; a max/p95 view shows magnitude but not exposed-makespan share; frontier accounting reports the share and routes to it. The cause was confirmed by configuration inspection.}
\label{tab:median-vs-frontier}
\small
\begin{tabularx}{\linewidth}{@{}L{0.27\linewidth}YL{0.22\linewidth}@{}}
\toprule
View & Callback stage reported as & Operator verdict \\
\midrule
Median (p50) dashboard & $\approx 2$\,ms ($0.13\%$ of the $1.526$\,s median step) & negligible \\
Max / p95 view & $7.17$\,s spike visible, but not its exposed-makespan share & magnitude only \\
\sys{} frontier & $24\%$ exposed-makespan share; single intermittent straggler rank & routed to callback \\
\bottomrule
\end{tabularx}
\end{table}

\subsection{Observed failure modes and negative cases}
\label{sec:failure-modes}
The cases above are positive routing outcomes; the complementary failure modes follow from a broad-stage, passive signal. The labeler gates (Table~\ref{tab:labeler-gates}) surface most of them as conservative labels rather than confident calls, but one is a confident misroute that the contract handles only with corroboration.
\begin{itemize}
\item \emph{Confident misroute on displaced work.} When forward CUDA work becomes host-visible downstream, the broad CPU-wall frontier confidently ranks the \emph{exposed} stage (backward) ahead of the origin stage (forward), yielding 0/10 top-1 in Section~\ref{sec:eval-routing}: a correct exposed-time call but a misleading origin, so forward/device is reported at top-2 with CUDA-event corroboration.
\item \emph{Unlocalized slowdowns.} Host CPU contention, a scheduling stall, or a filesystem hang can inflate total step time without concentrating in any one stage; \sys{} then reports a slow step with no confident stage, and node-list or GPU-utilization evidence is the next step.
\item \emph{A recurrent rank is not a node.} A persistent frontier-leader rank localizes a \emph{rank}, not a host; acting on it (e.g.\ draining a node) without rank-to-host mapping can target the wrong resource.
\item \emph{Threshold sensitivity and non-portable fixes.} Near a dominance or tie threshold, the leading stage can flip between windows, emitted as leader-switch or tie evidence; a correctly routed bottleneck can still lack a general remedy --- one field case routed correctly to a forward quantization cost with no low-overhead fix, and a data-prefetch change that recovered one workload left two others unchanged.
\end{itemize}
Apart from the forward/device misroute, no field diagnosis in Section~\ref{sec:cases} was overturned by a later heavy trace --- a consistency observation across the deployments, not a measured false-positive rate. The remaining cases downgrade to \texttt{telemetry\_limited}, \texttt{co\_critical}, or \texttt{role\_aware\_needed} rather than a single-stage call.

\section{Related Work}
\label{sec:related}
Prior approaches localize distributed-training slowdowns from rich inputs: trace and wait-state tools \emph{measure} wait, idle, and synchronization explicitly, while replay and causal methods \emph{infer} impact from event-trace critical paths, operator graphs, or virtual-speedup experiments. \sys{} needs neither: it infers the exposed-makespan attribution from a passive ordered vector of lumped stage durations --- no event trace, operator graph, or synchronized clock --- trading direct observability for an always-on signal cheap enough to leave running.

\paragraph{Framework profilers and trace analyzers.}
Framework profilers (PyTorch/TensorFlow Profiler, Nsight Systems~\cite{pytorchprofiler,tensorflowprofiler,nsightsystems}) and trace analyzers (HTA over Kineto~\cite{holistictraceanalysis}) give operator-, memory-, and CUDA-timeline resolution per trace and can attribute a stall causally by following a late rank back to the delaying event; that input is a captured trace too heavy to run continuously, so \sys{} sits one layer below, emitting a routing packet that says which window and stage class to ask a profiler about.

\paragraph{Instrumentation and straggler systems.}
Graph-replay systems (dPRO~\cite{hu2022dpro}, Chakra~\cite{sridharan2023chakra}) build a global dataflow graph and search optimizations, whereas \sys{} shows that ordered rank-local stage vectors alone --- not graphs, wait events, or synchronized clocks --- recover exact additive exposed-makespan accounting, small enough to leave on. Large-scale straggler platforms (FALCON~\cite{wu2024falcon}, GREYHOUND~\cite{wu2025greyhound}, MegaScale~\cite{jiang2024megascale}, what-if analyses~\cite{lin2025stragglers}) target $>$10\,000-GPU clusters and can ingest the \sys{} packet as a per-window signal alongside fleet state, while finer-grained Horovod/DeepSpeed and DDP/FSDP instrumentation~\cite{sergeev2018horovod,horovodtimeline,deepspeedcomm,deepspeedflops,li2020pytorchdistributed} exposes communication detail below the broad contract; Plumber~\cite{kuchnik2022plumber} diagnoses input-pipeline internals, whereas \sys{} only routes whether a data tail is exposed as synchronized step time and worth pulling such evidence on.

\paragraph{HPC wait-state and causal profiling.}
HPC wait-state and critical-path tools (Scalasca~\cite{geimer2008scalasca,geimer2010scalasca}, HPCToolkit~\cite{adhianto2010hpctoolkit}, TAU~\cite{shende2006tau}, wait-state/critical-path methods~\cite{boehme2010waitstates,boehme2012criticalpath}) attribute wait time on richer typed events than \sys{} collects, letting them separate independent co-critical paths from genuine wait; frontier accounting deliberately trades direct wait observation for an input of six floats per rank-step and a contract needing no event-tracing infrastructure, surfacing \texttt{co\_critical}, \texttt{role\_aware\_needed}, or \texttt{telemetry\_limited} where such a tool would split the path with side evidence. SyncPerf~\cite{alam2017syncperf} targets lock/condvar callsites in multithreaded programs. Causal profiling (Coz~\cite{curtsinger2015coz}) estimates impact via virtual speedups and tail-at-scale~\cite{dean2013tail} explains how rare component delays dominate; \sys{} borrows the goal of linking local slowdowns to throughput but injects no virtual speedups and claims no intervention impact from passive timers. MLPerf~\cite{mattson2020mlperf} measures workload-level performance, while \sys{} targets the layer below --- which stage, rank, role, or data path to investigate when a workload deviates from its expected throughput.

\section{Conclusion}
\label{sec:limitations}
Frontier accounting is a general identity: it applies to any training loop that reports an aligned ordered stage vector. We demonstrate it on homogeneous synchronous DDP transformer training through 128 ranks, with scoped FSDP \texttt{FULL\_SHARD} and ZeRO-1 checks; broader regimes are left to future evaluation. The one thing the signal does not do, by design, is assign root cause: it accounts for exposed time exactly and routes a heavier profiler to the window, stage, and rank to examine, and a frontier advance reads as a \emph{cause} only under the synchronization-wait model of Section~\ref{sec:sync-wait}, which the evidence labels flag where it does not hold.

The analytic core is unconditional: frontier advances telescope to the observed exposed step time for aligned stage vectors, and per-stage maxima and averages have formal $\min(R,S)$ overcounting and $R$ undercounting failure modes against it. The reference $O(RNS)$ window payload is small enough to run continuously and structured enough that heavier profiler collection fires only when the labels mark a window actionable --- a profiler-router that decides when, where, and for which rank a heavy trace is worth its cost.

\FloatBarrier
\clearpage
\appendix

\section{Stage Taxonomy and Ordered-Stage Contract}
\label{app:taxonomy}

\sys{} records a fixed list of non-overlapping trainer-stage spans on each rank using CPU wall-clock time (Table~\ref{tab:stages}). The forward/loss row may carry an optional side-channel metric \texttt{model.fwd\_loss\_\allowbreak cuda\_event\_ms}, a sampled device-stream elapsed time with a ready/missing state that is not an ordered stage and never enters the prefix vector, preserving exact exposed-makespan semantics.

\begin{table}[t]
\centering
\caption{Default broad stage taxonomy. Categories are intended for continuous accounting and diagnosis routing, not kernel-level attribution.}
\label{tab:stages}
\small
\begin{tabular}{@{}L{0.35\linewidth}L{0.61\linewidth}@{}}
\toprule
Stage & Meaning \\
\midrule
\texttt{data.next\_wait} & Main-process wait for the batch consumed by this logical step. \\
\texttt{model.fwd\_loss\_cpu\_wall} & Forward pass and loss construction as exposed to the host. \\
\texttt{model.backward\_cpu\_wall} & Backward pass and exposed distributed communication or synchronization. \\
\texttt{callbacks.cpu\_wall} & Callback, logging, checkpoint, and hook work around the step. \\
\texttt{optim.step\_cpu\_wall} & Optimizer step and related host-visible work. \\
\texttt{step.other\_cpu\_wall} & Residual: step wall time not covered by explicit spans. \\
\bottomrule
\end{tabular}
\end{table}

Let $w_{t,r}$ be the rank-local step wall time. The signed closure error is $e_{t,r}=w_{t,r}-\sum_{s\ne\mathrm{other}} d_{t,r,s}$, the residual stage is $d_{t,r,\mathrm{other}}=\max(0,e_{t,r})$, and the overlap error is $o_{t,r}=\max(0,-e_{t,r})$; large positive residual means the taxonomy is missing exposed work, large overlap error means spans are nested or double-counted, and both can downgrade a diagnosis to \texttt{telemetry\_limited}. Frontier accounting requires a common ordered boundary list within each diagnosis group: a stage may be broad but must be a contiguous, non-overlapping interval, and the instrumentation distinguishes \emph{ordered frontier stages} (in the prefix vector), \emph{side-channel probes} (nested but never in the prefix vector), and \emph{refined ordered schemas} (substages that replace a broad parent). Table~\ref{tab:contract-checks} lists the contract checks and conservative fallbacks. \sys{} also records a data wait on the logical step that consumes the batch, not the loop iteration that called \texttt{next(dataloader)}, so a step-$t{+}1$ data stall is not shifted into step-$t$ compute.

\begin{table}[t]
\centering
\caption{Ordered-stage contract checks and conservative fallback behavior.}
\label{tab:contract-checks}
\small
\begin{tabularx}{\linewidth}{@{}L{0.35\linewidth}Y@{}}
\toprule
Check & Fallback if violated \\
\midrule
Only one ordered frontier stage active per rank & Reject nested ordered spans; keep nested measurements only when declared \texttt{side\_channel=true}. \\
Common schema version, ordered stage list, and stage-order hash inside a diagnosis group & Close the current window and emit \texttt{telemetry\_limited}; do not merge mismatched rows. \\
All ranks in the diagnosis group present at the window boundary & Emit local summaries if safe, but downgrade distributed labels to \texttt{telemetry\_limited}. \\
Residual closure and overlap error within configured thresholds & Keep \texttt{frontier\_accounting} if the vector is usable, but suppress stronger diagnosis labels. \\
Role metadata sufficient for the chosen group & Emit \texttt{role\_aware\_needed} rather than applying a global frontier across incompatible roles. \\
\bottomrule
\end{tabularx}
\end{table}

\FloatBarrier
\section{Full Diagnosis Label Table}
\label{app:full-labels}

Table~\ref{tab:labels-full} lists every label emitted by the deployed labeler, including the forward-event variants and the gradient-accumulation ambiguity label.

\begin{table}[t]
\centering
\caption{Full diagnosis labels and downgrade triggers.}
\label{tab:labels-full}
\small
\begin{tabularx}{\linewidth}{@{}L{0.30\linewidth}L{0.10\linewidth}YL{0.18\linewidth}@{}}
\toprule
Label & Axis & Meaning & Downgrade trigger \\
\midrule
\texttt{frontier\_accounting} & Accounting & Additive exposed-makespan decomposition only. & None; this is the base claim. \\
\texttt{likely\_sync\_wait} & Model fit & Frontier, lag, and workload semantics support upstream wait-induced delay. & Weak role metadata, high closure error, or low dominance margin. \\
\texttt{sync\_wait\_dependent} & Model fit & Frontier share is high but static direct gain is low; actionability depends on the wait model. & Direct evidence or traces show independent productive work. \\
\texttt{direct\_exposure} & Direct gain & Raw duration, spread, and clipped static sensitivity agree with the frontier stage. & Static gain small relative to frontier share. \\
\texttt{forward\_device\_supported} & Device evidence & Forward/loss frontier evidence is accompanied by high sampled CUDA-event forward time. & Event samples are missing, low, or scope-limited. \\
\texttt{forward\_spillover\_suspected} & Device evidence & Forward CUDA-event time is high, but host-visible exposed cost appears later, often in backward. & Profiler or stream evidence places the work elsewhere. \\
\texttt{forward\_host\allowbreak\_overhead\allowbreak\_suspected} & Device evidence & CPU-wall forward/loss is high while sampled CUDA-event forward time is low. & Event scope misses side-stream work or host evidence points elsewhere. \\
\texttt{forward\_event\allowbreak\_scope\allowbreak\_limited} & Device evidence & Event ready ratio is low or the trainer uses streams not covered by the event markers. & Stream coverage or readiness improves. \\
\texttt{co\_critical} & Ambiguity & Multiple stages or ranks can plausibly remain bottlenecks after optimizing one stage. & Side evidence separates the paths. \\
\texttt{gradient\_accumulation\allowbreak\_ambiguous} & Ambiguity & Accumulation microsteps were collapsed or mixed; data/backward displacement cannot be separated. & Collect unfolded accumulation-indexed substages. \\
\texttt{role\_aware\_needed} & Role clarity & Rank roles differ enough that global rank aggregation is unsafe. & Role grouping becomes available. \\
\texttt{telemetry\_limited} & Telemetry & Residuals, overlap error, gather failures, metadata gaps, or missing probes cap confidence. & Telemetry quality is restored. \\
\bottomrule
\end{tabularx}
\end{table}

\FloatBarrier
\section{Labeler Default Gates}
\label{app:labeler-gates}

The labeler is deterministic given the stage matrix, schema metadata, optional side evidence, and threshold configuration: it validates the ordered-stage contract and schema/world membership, computes prefixes, frontier advances, shares, and the routing set, computes lag/delta-lag/tie/leader-switch evidence and clipped direct-exposure gain, applies telemetry-quality and role-aware gates, evaluates optional CUDA-event or communication side evidence, and emits labels, the routing set, the ambiguity evidence set, and downgrade reasons. Table~\ref{tab:labeler-gates} lists the default gates. The co-criticality logic uses an ambiguity set $E_{\mathrm{amb}}=C_A\cup C_G$ (top stages by frontier share and by clipped static gain): a window is \texttt{sync\_wait\_dependent} when $A_{s_1}>\gamma_A$, $G_{s_1}<\gamma_G$, and a caller-supplied model-fit indicator $W_{s_1}=1$; the same low-gain configuration with $W_{s_1}=0$ is \texttt{co\_critical}, and near-tied or high-leader-switch windows are likewise downgraded. The artifact labeler uses the safe default $W_s=0$.

\begin{table}[H]
\centering
\caption{Default labeler gates used as conservative starting points.}
\label{tab:labeler-gates}
\small
\begin{tabularx}{\linewidth}{@{}L{0.35\linewidth}L{0.18\linewidth}Y@{}}
\toprule
Gate & Default & Purpose \\
\midrule
Closure residual share & $\le 0.05$ & Suppress strong labels when explicit spans do not close the step. \\
Overlap error share & $\le 0.01$ & Detect nested or double-counted ordered stages. \\
Missing-rank count & $0$ & Require complete distributed evidence for distributed labels. \\
Event-ready ratio & $\ge 0.8$ & Permit forward-event support labels only when enough sampled pairs complete. \\
Minimum event samples & $\ge 5$ & Avoid device-evidence labels from a single ready event. \\
Frontier-share dominance & $\gamma_A=0.4$ & Require a leading exposed-makespan share for strong stage labels. \\
Static-gain threshold & $\gamma_G=0.1$ & Separate direct exposure from wait-dependent frontier share. \\
Share/gain tie tolerance & $\eta_A=\eta_G=0.05$ & Emit co-critical evidence for near ties. \\
Leader-switch tolerance & $\eta_Q$, $\gamma_{\mathrm{switch}}$, $\gamma_{\mathrm{elig}}$ & Count switches only between confident unique leaders. \\
Candidate cumulative threshold & $\tau_C=0.80$ & Define compact candidate routing set. \\
Model-fit indicator & default $W_s=0$ & Do not infer sync-wait dependence without workload or side evidence. \\
\bottomrule
\end{tabularx}
\end{table}

\FloatBarrier
\section{Frontier Bounds and Measurement-Error Stability}
\label{app:frontier-bounds-proofs}
\label{app:measurement-error}

Notation matches Section~\ref{sec:method}.

\begin{proof}[Proof of Proposition~\ref{prop:max-overcount}]
Let $u_s=\argmaxop_r d_{t,r,s}$ and regroup the stage maxima by maximizing rank: $M_t=\sum_{r}\sum_{s:u_s=r} d_{t,r,s}$. For each rank $r$, $\sum_{s:u_s=r} d_{t,r,s}\le\sum_s d_{t,r,s}\le F_{t,S}$, so summing over the at most $R$ ranks gives $M_t\le R\,F_{t,S}$; the trivial per-stage bound gives $M_t\le S\,F_{t,S}$, hence $M_t\le\min(R,S)\,F_{t,S}$. Also $F_{t,S}=\max_r\sum_s d_{t,r,s}\le\sum_s\max_r d_{t,r,s}=M_t$. The ratio $\min(R,S)$ is reached by assigning duration $D$ to each of $\min(R,S)$ distinct rank-stage pairs and zero elsewhere, with no rank receiving more than one nonzero duration.
\end{proof}

\begin{proof}[Proof of Proposition~\ref{prop:avg-underreport}]
Since each rank total is at most $F_{t,S}$, the average rank total is at most $F_{t,S}$; since the maximum rank total is at most the sum of rank totals, the average is at least $F_{t,S}/R$. The lower bound is reached when one rank has total $D$ and all others zero.
\end{proof}

\begin{proposition}[Measurement-error stability]
If every measured stage duration has absolute error at most $\epsilon$, then each prefix frontier has error at most $s\epsilon$ and each frontier advance has error at most $(2s-1)\epsilon$. Window shares are stable when the leading-stage margin is large compared with these errors and the denominator $\sum_{t}F_{t,S}$ is bounded away from zero.
\end{proposition}
\begin{proof}
Each prefix sums at most $s$ perturbed durations, so its error is at most $s\epsilon$, and the max of values changes by no more than the maximum perturbation, so $F_{t,s}$ changes by at most $s\epsilon$. Since $a_{t,s}=F_{t,s}-F_{t,s-1}$, its error is at most $(2s-1)\epsilon$ by the triangle inequality. Normalized shares divide by the window denominator; when that is near zero the implementation reports raw advances or downgrades percentage labels.
\end{proof}

\FloatBarrier
\section{Extended Routing Results}
\label{app:extended-eval}

\begin{table}[H]
\centering
\caption{Full hidden-rank routing summary. E3 rows use DDP without gradient accumulation and five seeds per rank count; callback rows use a three-seed synchronization study. The daggered callback/host row is a host-only scope-boundary control with no immediate barrier. Delay/p50 normalizes by the median no-fault step p50 for the matching rank count.}
\label{tab:stagefrontier-final-top12}
\scriptsize
\begin{tabular}{@{}llrrrrrrr@{}}
\toprule
Scenario & Ranks & Delay & Delay/p50 & Rows & Top-1 match & Top-2 routing & Cand.\ hit & Cand.\ size \\
\midrule
Data & 8 & 120 ms & 0.58 & 5 & 5/5 & 5/5 & 5/5 & 2.0 \\
Data & 32 & 120 ms & 0.51 & 5 & 5/5 & 5/5 & 5/5 & 2.0 \\
Backward & 8 & 120 ms & 0.58 & 5 & 5/5 & 5/5 & 5/5 & 2.0 \\
Backward & 32 & 120 ms & 0.51 & 5 & 5/5 & 5/5 & 5/5 & 2.0 \\
Backward/comm & 8 & 120 ms & 0.58 & 5 & 5/5 & 5/5 & 5/5 & 2.0 \\
Backward/comm & 32 & 120 ms & 0.51 & 5 & 5/5 & 5/5 & 5/5 & 2.0 \\
Forward/device & 8 & 120 ms & 0.58 & 5 & 0/5 & 5/5 & 5/5 & 2.0 \\
Forward/device & 32 & 120 ms & 0.51 & 5 & 0/5 & 5/5 & 5/5 & 2.0 \\
Forward/host & 8 & 120 ms & 0.58 & 5 & 5/5 & 5/5 & 5/5 & 2.0 \\
Forward/host & 32 & 120 ms & 0.51 & 5 & 5/5 & 5/5 & 5/5 & 2.0 \\
Callback/sync & 8 & 120 ms & 0.58 & 3 & 0/3 & 3/3 & 3/3 & 2.0 \\
Callback/host$^\dagger$ & 8 & 120 ms & 0.58 & 3 & 0/3 & 0/3 & 0/3 & 2.0 \\
\bottomrule
\end{tabular}
\end{table}

\begin{table}[t]
\centering
\caption{Sensitivity of the E3 \sys{} routing candidate set to the cumulative threshold $\tau_C$, recomputed from the stored stage scores for the same 50 rows as Table~\ref{tab:routing-and-baselines}. Higher thresholds preserve candidate hit but reduce compactness.}
\label{tab:tau-sensitivity}
\small
\begin{tabular}{@{}rrrr@{}}
\toprule
$\tau_C$ & Cand.\ hit & Avg cand.\ size & Max cand.\ size \\
\midrule
0.70 & 50/50 & 2.00 & 2 \\
0.75 & 50/50 & 2.00 & 2 \\
0.80 & 50/50 & 2.00 & 2 \\
0.85 & 50/50 & 2.20 & 3 \\
0.90 & 50/50 & 3.00 & 3 \\
\bottomrule
\end{tabular}
\end{table}

\FloatBarrier
\section{Artifact Reproducibility}
\label{app:artifact-checklist}

The evaluation is backed by the \path{stagefrontier-artifact} repository. Validation-cluster rows use the NVIDIA PyTorch 24.12 container (PyTorch 2.6.0a0, CUDA 12.6), NCCL 2.23.4 for training, Gloo telemetry gather for the main routing matrix (Gloo and NCCL for the overhead-scale matrix), bf16 transformer runs, deterministic seeds recorded in each structured row, and the default labeler gates of Table~\ref{tab:labeler-gates}. E8 additionally records the sharding wrapper, FSDP strategy, ZeRO stage, and CUDA memory counters; E9 records the selected-window tool mode, capture mode, Kineto/Nsight artifact sizes, and HTA ingestion status. Local CPU smoke tests exercise the command surface but do not regenerate paper tables; E0--E9 paper rows require GPU Slurm allocations. The artifact's \texttt{paper\_outputs} snapshots archive the validated structured summary files and provenance hashes for the CSV tables, and an \texttt{artifact\_manifest.csv} maps each evaluation table to its source artifact and regeneration command (checked-in Slurm wrappers for E0--E9 and the callback synchronization study).

\begin{table}[t]
\centering
\caption{Experiment groups referenced by the evaluation. All groups except E7 disable gradient accumulation; E9 is a selected-window profiler comparison, not a continuous deployment.}
\label{tab:experiment-groups}
\footnotesize
\begin{tabularx}{\linewidth}{@{}L{0.09\linewidth}YL{0.32\linewidth}@{}}
\toprule
Group & Configuration & Role \\
\midrule
E0 & 8 ranks, no fault, $q\in\{0,0.05\}$, 600 steps & Telemetry completeness; no stray sleeps or CUDA syncs \\
E1 & 8--128 ranks, seeds 0--4, Gloo and NCCL gather, paired logger-off / CPU-wall / $q{=}0.05$ rows & Paired-bootstrap overhead through 128 ranks \\
E2 & 8 ranks, seeds 0--2, 120\,ms data, comm, and forward injections & Small-scale routing pilot \\
E3 & 8 and 32 ranks, seeds 0--4, 12/30/60/120\,ms hidden-rank delays & Detectability curve; top-1/top-2 routing \\
E4 & 8 and 32 ranks, profiler spot checks, 150/180\,ms data-tail & Separates candidate-boundary effects from data-stage attribution \\
E5 & 8 and 32 ranks, $q\in\{0,0.05,1\}$ CUDA-event side channel & Forward-device vs.\ forward-host side evidence \\
Callback & 8 ranks, seeds 0--2, 60--240\,ms sync-bearing and 120\,ms host-only rows & Sync vs.\ host-only callback; top-1/top-2 boundary \\
E6 & 8 ranks, seeds 0--2, A/B/A with a 120\,ms sync callback in B only & Removed-injection consistency \\
E7 & 8 ranks, seeds 0--4, accumulation factor 4 with DDP \texttt{no\_sync} & Gradient-accumulation ordered-stage check \\
E8 & 8/16/32 ranks, seeds 0--2, FSDP \texttt{FULL\_SHARD} and \texttt{ZeroRedundancyOptimizer} ZeRO-1, 180\,ms faults & Sharded-data-parallel scope spot check \\
E9 & 32 ranks, seeds 0--2, 40-step capture (inner 20 scored), 180\,ms faults & \sys{} vs.\ PyTorch Profiler, HTA, and Nsight cost \\
Scale & 64/128 ranks, 120\,ms comm and 120--360\,ms data-tail faults & Routing persistence beyond 8/32 ranks \\
\bottomrule
\end{tabularx}
\end{table}

\FloatBarrier
\bibliographystyle{abbrvnat}
\bibliography{references}

\end{document}